\documentclass[graybox]{svmult}

\usepackage{array}
\usepackage{mathptmx}       
\usepackage{helvet}         
\usepackage{courier}        
\usepackage{type1cm}        
\usepackage{makeidx}         
\usepackage{graphicx}        
\usepackage{multicol}        
\usepackage[bottom]{footmisc}

\makeindex             

\usepackage{amsgen,amsmath,amstext,amsbsy,amsopn,amsfonts, amssymb,stmaryrd}
\usepackage{tcolorbox}
\usepackage{fancyhdr}
\usepackage[english]{babel}
\usepackage{eurosym}
\usepackage{algorithm,xspace}
\usepackage[noend]{algorithmic}
\usepackage{tikz} 
\usetikzlibrary{chains,positioning,scopes,shapes} 

\usepackage{subfigure}
\usepackage{color}

\newcommand{\matN}{\mathbb N}

\newcommand{\remove}[1]{}

\begin{document}

\title*{On List $k$-Coloring Convex Bipartite Graphs} 
\author{Josep D\'{\i}az,  \"{O}znur Ya\c{s}ar Diner,  Maria Serna,  and  Oriol Serra}
\institute{Josep D\'{\i}az \at ALBCOM Research Group, CS Department, Universitat Polit\`{e}cnica de Catalunya, \email{diaz@cs.upc.edu}
 \and
 \"{O}znur Ya\c{s}ar Diner \at Mathematics Department, Universitat Polit\`{e}cnica de Catalunya and 
 Computer Engineering Department, Kadir Has University, Istanbul,  \email{oznur.yasar@khas.edu.tr} 
 \and
  Maria Serna \at ALBCOM Research Group, CS Department, Universitat Polit\`{e}cnica de Catalunya, \email{mjserna@cs.upc.edu}
\and
 Oriol Serra \at  Mathematics Department, Universitat Polit\`{e}cnica de Catalunya,
 \email{oriol.serra@upc.edu} 
 }

\maketitle

\abstract{List $k$--Coloring ({\sc Li $k$-Col}) is the decision problem  asking if a given graph admits a proper coloring compatible with a given list assignment to its vertices with colors in  $\{1,2,\ldots ,k\}$. The problem is known to be NP-hard even for $k=3$ within the class of $3$--regular planar bipartite graphs and for $k=4$ within the class of chordal bipartite graphs.  In 2015 Huang, Johnson and Paulusma   asked for the complexity of {\sc Li $3$-Col} in the class of chordal  bipartite graphs. In this paper we give a partial answer to this question  by showing that {\sc Li $k$-Col} is polynomial in the class of convex bipartite graphs.  We show first that biconvex bipartite graphs admit a multichain ordering, extending the classes of graphs where a polynomial algorithm of Enright, Stewart and Tardos (2014) can be applied to the problem. We provide  a  dynamic programming algorithm to solve the {\sc Li $k$-Col} in the calss of convex bipartite graphs. Finally we show how our algorithm can be modified to solve the more general {\sc Li $H$-Col} problem on convex bipartite graphs. }


\section{Introduction}\label{sec:intro}

A \emph{coloring} of a a graph $G=(V, E)$ is a map $c:V\to \matN$.  A coloring is \emph{proper} if no two adjacent vertices are assigned the same color.  
 If there is a proper coloring of a graph that uses at most $k$ colors then we say that  $G$ is $k$-\emph{colorable} and that $c$ is a $k$-\emph{coloring} for $G$.
 The coloring problem {\sc Col} asks, for a given  graph $G=(V,E)$ and a positive integer $k$, whether there is a $k$-coloring for $G$. When $k$ is fixed, we have the {\sc $k$-Coloring} problem. 
 
A list assignment $L:V\to 2^{\matN}$ is a map assigning a set of positive integers to each vertex of $G$. Given $G$ and $L$, the List Coloring problem {\sc LiCol} asks for the existence of a proper coloring $c$ that obeys $L$, i.e., each vertex receives a color from its own list. If the answer is positive, $G$ is said to be $L$-\emph{colorable}.
 Variants of the problem are defined by bounding the total number of available colors or by bounding the list size. In {\sc List} $k$-{\sc Coloring} ({\sc Li $k$-Col}), $L(v)\subseteq  \{1,2,\ldots,k\}$ for each $v\in V$. Thus, there are $k$ colors in total. On the other hand in $k$-{\sc List Coloring} ({\sc $k$-LiCol}) each list $L$ has size at most $k$, in this case the total number of colors can be larger than $k$. 
 
 
Precoloring Extension, {\sc PrExt}, is a special case of {\sc LiCol} and a generalization of {\sc Col}. In {\sc PrExt} all of the vertices in a subset $W$ of $V$ are previously colored and the task is to extend this coloring to all of the vertices. If, in addition,
the total number of colors is bounded, say by $k$, then it is called the k-Precoloring Extension, {\sc $k$-PrExt}. $k$-{\sc Col} is clearly a special case of {\sc $k$-PrExt}, which in turn is a special case of {\sc Li $k$-Col}. Refer to \cite{golovach2016} for a chart summarizing these relationships. 

For general graphs {\sc Col} and its variants {\sc LiCol} and {\sc PrExt} are NP--complete; see \cite{karp1972, garey1979}. Most of their variants are NP-complete even when the parameter $k$ is fixed for small values of $k$: $k$-{\sc Col}, $k$-{\sc LiCol}, {\sc Li $k$-Col} and {\sc $k$-PrExt} are NP-complete when $k\geq 3$ \cite{lovasz1973} and they are polynomially solvable when $k\leq 2$ \cite{erdos1979, vizing1976}. 

Concerning the complexity of these problems in graph classes, {\sc Col} is solvable in polynomial time for perfect graphs \cite{grostchel} whereas {\sc LiCol} is NP-complete when restricted to perfect graphs and many of its subclasses, such as split graphs, bipartite graphs \cite{kubale1992} and interval graphs \cite{biro1992}. On the other hand {\sc LiCol} is polynomially solvable for trees, complete graphs and graphs of bounded treewidth \cite{jansen1997}. Refer to Tuza \cite{tuza1997} and more recently to Paulusma \cite{paulasma2016} for related surveys.

For small values of $k$, Jansen and Scheffler \cite{jansen1997} have shown that $3$--{\sc LiCol} is NP-complete when restricted to complete bipartite graphs and cographs, as observed in \cite{golovach2014}. Kratochv\'il and Tuza \cite{kratochvil1994} showed that $3$--{\sc LiCol} is NP-complete even if  each color appears in at most three lists, each vertex in the graph has degree at most three and the graph is planar. {\sc $3$-PrExt} is NP-complete even for $3$--regular planar bipartite graphs and for planar bipartite graphs with maximum degree $4$ \cite{chlebik}. 

For fixed $k\geq 3$, {\sc Li $k$-Col} is 
polynomially solvable for $P_5$-free graphs \cite{hoang2010}.
Note that chordal bipartite graphs contain $P_5$-free graphs but $P_6$ free graphs are incomparable with chordal bipartite graphs \cite{spinrad}. {\sc Li $3$-Col} is polynomial for $P_6$-free graphs \cite{broersma2013} and for $P_7$-free graphs \cite{bonomo2018}. Computational complexity of {\sc Li $3$-Col} for $P_8$-free bipartite graphs is open \cite{bonomo2018}. Even the restricted case of {\sc Li $3$-Col} for $P_8$-free chordal bipartite graphs is open. Golovach et. al. \cite{golovach2016} give a survey that  summarizes the results for {\sc Li $k$-Col} on $H$-free graphs in terms of the structure of $H$.

{\sc PrExt} problem is solvable in linear time on $P_5$-free graphs and it is NP-complete when restricted to $P_6$-free chordal bipartite graphs \cite{hujter1996}. {\sc $3$-PrExt} is NP-complete even for planar bipartite graphs \cite{kratochvil1993}, even for those having maximum degree 4 \cite{chlebik}. Recall that {\sc PrExt} generalizes $k$-PrExt and {\sc Li $k$-Col} generalizes {\sc $k$-PrExt}. But there is no direct relation between {\sc PrExt} and {\sc Li $k$-Col} \cite{golovach2016}.

Coloring problems can be placed in the more general class of $H$--coloring problems. Given two graphs $G$ and $H$, a function $f : V(G) \rightarrow V(H)$ such that $f(u)$ and $f(v)$ are adjacent in $H$ whenever $u$ and $v$ are
adjacent in $G$ is called a graph homomorphism from $G$ to $H$. For a fixed graph $H$ and for an input $G$, the $H$-coloring problem, {\sc $H$-Col}, asks whether there is a $G$ to $H$ homomorphism. In the list $H$-coloring problem, {\sc Li $H$-Col}, each vertex of the input graph $G$ is associated with a list of vertices of $H$ and the question is whether a $G$ to $H$ homomorphism exists that maps each vertex to a member of its list. Observe that {\sc Li $H$-Col} is a generalization of {\sc Li $k$-Col}. The complexities of
the $H$--coloring and list $H$--coloring problems for arbitrary input graphs are completely
characterized in terms of the structure of $H$, see Ne\v set\v ril and Hell \cite{nesetril}.

Although intensive research on this subject has been undertaken in the last two decades, there are still numerous open questions regarding computational complexities on {\sc LiCol} and its variants when they are restricted to certain graph classes. Huang, Johnson and Paulusma \cite{huang2015} proved that {\sc Li $4$-Col} is NP-complete for $P_8$--free chordal bipartite graphs and $4$--{\sc PrExt} is NP-complete for $P_{10}$--free chordal bipartite graphs. They further pose the problem on the computational complexity of the {\sc Li $3$-Col} and {\sc $3$-PrExt} on chordal bipartite graphs. Here {\sc Li $k$-Col} and {\sc $k$-PrExt} on convex bipartite graphs, a proper subclass of chordal bipartite graphs, are studied and a partial answer to this question is given. Figure \ref{chart}  summarizes the related results.

 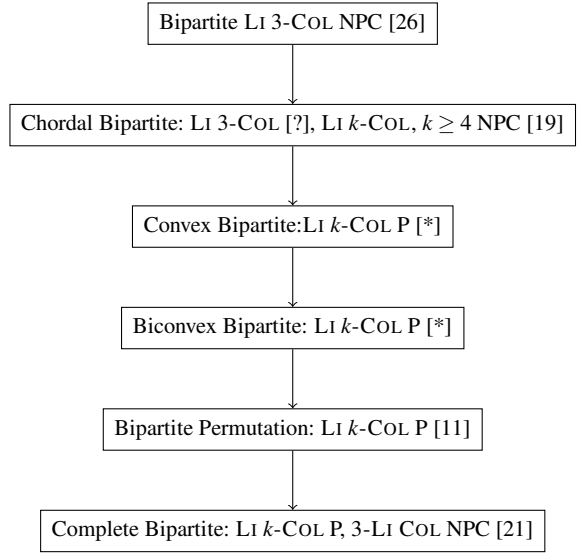
\begin{figure}
  \begin{center}
  \begin{tikzpicture}[scale=0.9]
  \node[draw] (CB) at (2,2) {
  \begin{tabular}{c}Complete Bipartite: {\sc Li $k$-Col} P, {\sc $3$-Li Col} NPC  \cite{jansen1997}\end{tabular}};
 \node[draw] (BP) at (2,3.5) {
  \begin{tabular}{c} Bipartite Permutation: {\sc Li $k$-Col} P  \cite{enright2014}\end{tabular}};
   \node[draw] (BC) at (2,5) {
  \begin{tabular}{c} Biconvex Bipartite: {\sc Li $k$-Col} P  [*]\end{tabular}};
  
  \node[draw] (CxB) at (2,6.5) {
  \begin{tabular}{c} Convex Bipartite:{\sc Li $k$-Col} P [*]\end{tabular}};
  
  \node[draw] (ChB) at (2,8) {
  \begin{tabular}{c} Chordal Bipartite: {\sc Li $3$-Col} [?], {\sc Li $k$-Col, $k\geq 4$} NPC \cite{huang2015}\end{tabular}};
  
  \node[draw] (PB) at (2,9.5) {
  \begin{tabular}{c} Bipartite {\sc Li $3$-Col} NPC   \cite{kubale1992}  \end{tabular}};
  
  \path (BP) edge[->]   (CB);
  \path (BP) edge[<-] (BC);
  \path (BC) edge[<-] (CxB);
  \path (CxB) edge[<-] (ChB);
  \path (ChB) edge[<-] (PB);
  \end{tikzpicture}
 \end{center}
  \caption{Chart for known complexities for {\sc LiCol} and its variants for chordal bipartite graphs and its subclasses, for $k\geq 3$. The complexity results marked with [*] is the topic of this paper, while [?] stands for open cases. Results without reference are trivial. P stands for Polynomial and NPC for NP-complete.}
  \label{chart}
 
  \end{figure}

A  bipartite graph $G=(X\cup Y, E)$ is convex if it admits an ordering on one of the parts of the bipartition, say $X$,  such  that the neighbours of each vertex in $Y$ are consecutive in this order. If both color classes admit such an ordering the graph is called {biconvex bipartite} (see Section \ref{sec:prelim} for  formal definitions). Chordal bipartite graphs contain convex bipartite graphs properly. Convex bipartite graphs contain as a proper subclass biconvex bipartite graphs which contain bipartite permutation graphs properly. More information on these classes can be found in Spinrad \cite{spinrad} and in Brandst\"{a}dt, Le and  Spinrad \cite{brandstadt}.
 
Enright, Stewart and Tardos \cite{enright2014} have shown that {\sc Li $k$-Col} is solvable in polynomial time when restricted to graphs with all connected induced subgraphs having a multichain ordering. They apply this result to  permutation graphs and  interval graphs. Here we show that connected biconvex graphs also admit a multichain ordering, implying a  polynomial time algorithm for {\sc Li $k$-Col} on this graph class.  

From the point of view of parameterized complexity, treewidth can be computed in polynomial time on chordal bipartite graphs \cite{kloks1993}. {\sc Li $k$-Col} can be solved in polynomial time  on chordal bipartite graphs with bounded treewidth \cite{jansen1997, diaz2002} which includes chordal bipartite graphs of bounded degreee \cite{lozin}. {\sc Li $k$-Col} is polynomial for graphs of bounded cliquewidth \cite{courcelle}. Note that convex bipartite graph contains graphs with unbounded treewidth as well as graphs with unbounded cliquewidth. 

The paper is organized as follows. In Section \ref{sec:prelim} we give the necessary definitions. In Section \ref{sec:biconvex} we show that connected biconvex bipartite graphs admit  multichain ordering. In Section \ref{sec:convex}, we show that {\sc Li $k$-Col} is polynomially solvable when it is restricted to convex bipartite graphs. Then we show how to extend this result to {\sc Li $H$-Col}.  


\section{Preliminaries}\label{sec:prelim}

We consider finite simple graphs $G=(V,E)$. For terminology refer to Diestel \cite{diestel}.

%
%

An edge joining non adjacent vertices in the cycle, $C_n$, is called a {\it chord}. A graph $G$ is {\it chordal} if every induced cycle of length $n\geq 4$ has a chord. {\it Chordal bipartite graphs} are bipartite graphs in which every induced $C_n, n\geq 6$ has a chord. This graph class is introduced by Golumbic and Gross \cite{golumbic1978}. Chordal bipartite graphs may contain induced $C_4$, so they do not constitute a subclass of chordal graphs but it is a proper subclass of bipartite graphs. Chordal bipartite graphs can be recognized in polynomial time \cite{paige1987}.

A bipartite graph is represented by $G = ( X\cup Y, E)$, where $X$, $Y$ form a bipartition of the vertex set into stable sets. An ordering of the vertices $X$ in a bipartite graph $G = ( X\cup Y, E)$ has the {\it adjacency property} (or the ordering is said to be {\it convex}) and $G$ is said to have {\it convexity with respect to $X$} if, for each vertex $v \in Y$,  $N(v)$ consists of vertices which are consecutive in the ordering of $X$. We say that an ordering of the vertices $X$ in a bipartite graph $G = ( X\cup Y, E)$ has the {\it enclosure property} if for every pair of vertices $u, v \in Y$ such that $N(u) \subseteq N(v)$, the vertices in $N(v)\backslash N(u)$ occur consecutively in the ordering of $X.$

\begin{figure}[t]
\begin{center}
\begin{tikzpicture}[scale=0.8,dot/.style={draw,circle,minimum size=1mm,inner sep=0pt,outer sep=0pt}]

\draw (2.5,2.5) node {$y_1$};
\draw (5.5,2.5) node {$y_2$};
\draw (7.5,2.5) node {$y_3$};
\draw (9,2.5) node {$y_4$};

\draw (0,2.5) node {$Y$};
\draw (0,0) node {$X$};

\coordinate [dot,fill =black] (y1) at (2.5,2);
\coordinate [dot,fill =black] (y2) at (5.5,2);
\coordinate [dot,fill =black](y3) at (7.5,2);
\coordinate [dot,fill =black] (y4) at (9,2);

\coordinate [dot,fill =black] (x1) at (1,0.5);
\coordinate [dot,fill =black] (x2) at (2,0.5);
\coordinate [dot,fill =black](x3) at (3,0.5);
\coordinate [dot,fill =black] (x4) at (4,0.5);
\coordinate [dot,fill =black] (x5) at (5,0.5);
\coordinate [dot,fill =black] (x6) at (6,0.5);
\coordinate [dot,fill =black](x7) at (7,0.5);
\coordinate [dot,fill =black](x8) at (8,0.5);
\coordinate [dot,fill =black](x9) at (9,0.5);

\draw (1,0) node {$x_1$};
\draw (2,0) node {$x_2$};
\draw (3,0) node {$x_3$};
\draw (4,0) node {$x_4$};
\draw (5,0) node {$x_5$};
\draw (6,0) node {$x_6$};
\draw (7,0) node {$x_7$};
\draw (8,0) node {$x_8$};
\draw (9,0) node {$x_9$};

\path (y1) edge (x1);
\path (y1) edge (x2);
\path (y1) edge (x3);
\path (y1) edge (x4);
\path (y1) edge (x5);

\path (y2) edge (x4);
\path (y2) edge (x5);
\path (y2) edge (x6);
\path (y2) edge (x7);

\path (y3) edge (x7);
\path (y3) edge (x8);

\path (y4) edge (x5);
\path (y4) edge (x6);
\path (y4) edge (x7);
\path (y4) edge (x8);
\path (y4) edge (x9);

\end{tikzpicture}
\end{center}
\caption{A convex bipartite graph which is not biconvex.\label{fig:convex-bgraph}}
\end{figure}
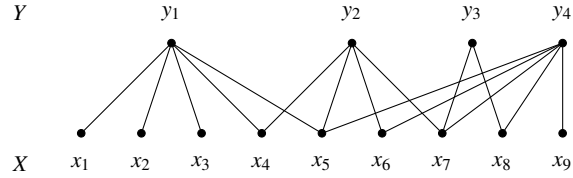

{\it Convex bipartite graphs} are bipartite graphs $G = ( X\cup Y, E)$ that have the adjacency property on one of the partite sets and {\it biconvex bipartite graphs} have the the adjacency property on both partite sets $X$ and $Y$. Fig. ~\ref{fig:convex-bgraph} shows a graph that is convex but not biconvex. {\it Bipartite permutation graphs} are biconvex bipartite graphs in which one of the partite sets obeys both the adjacency and the enclosure properties. There are linear time recognition algorithms for these classes \cite{spinrad1987, nussbaum2010}. 

A {\it chain graph}  is a bipartite graph that contains no induced $2K_2$ (a graph formed by two independent edges)  \cite{yannakakis1982}. The following characterization from \cite{enright2014} is equivalent: a connected bipartite graph with bipartite sets $X$ and $Y$ is a chain graph if and only if for any two vertices $y_1, y_2 \in Y$ we have  $N(y_1) \subseteq N(y_2)$ or $N(y_2)\subseteq N(y_1)$. If the vertices in $X$ are ordered with respect to their degrees starting from the highest degree, then for any $y\in Y$, the vertices in $ N(y)$ will be consecutive in the ordering on $X$ and, if the graph is connected, there is always a vertex $y\in Y$ so that   $N(y)$ includes the first vertex in $X$. In particular, chain graphs are a proper subclass of convex bipartite graphs. 

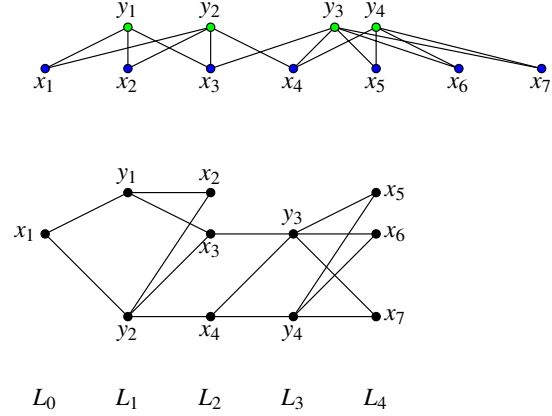
\begin{figure}[t] 
\begin{center}
\scalebox{1.1}{%
\begin{tikzpicture}[dot/.style={draw,circle,minimum size=1mm,inner sep=0pt,outer sep=0pt}]
\coordinate [dot,fill = blue] (X1) at (0,3);
\coordinate [dot,fill = blue] (X2) at (1,3);
 \coordinate [dot,fill = blue] (X3) at (2,3);
 \coordinate  [dot,fill = blue](X4) at (3,3);
\coordinate  [dot,fill = blue](X5) at (4,3);
\coordinate [dot,fill = blue] (X6) at (5,3);
\coordinate [dot,fill = blue] (X7) at (6,3);
\coordinate[dot,fill = green] (Y1) at (1,3.5);
\coordinate [dot,fill = green] (Y2) at (2,3.5);
\coordinate [dot,fill = green] (Y3) at (3.5,3.5);
\coordinate [dot,fill = green]  (Y4) at (4,3.5);
\coordinate [dot,fill = black] (LX1) at (0,1);
\coordinate [dot,fill = black]  (LX2) at (2,1.5);
\coordinate [dot,fill = black]  (LX3) at (2,1);
\coordinate [dot,fill = black]  (LX4) at (2,0);
\coordinate [dot,fill = black]  (LX5) at (4,1.5);
\coordinate [dot,fill = black]  (LX6) at (4, 1);
\coordinate [dot,fill = black]  (LX7) at (4, 0);
\coordinate [dot,fill = black]  (LY1) at (1,1.5);
\coordinate [dot,fill = black]  (LY2) at (1, 0);
\coordinate [dot,fill = black]  (LY3) at (3,1);
\coordinate [dot,fill = black]  (LY4) at (3, 0);

\draw (0,-1) node {$L_0$};
\draw (1,-1) node {$L_1$};
\draw (2,-1) node {$L_2$};
\draw (3,-1) node {$L_3$};
\draw (4,-1) node {$L_4$};

\draw [black] (X1) -- (Y1)-- (X2) -- (Y2)-- (X3) -- (Y3) -- (X4)-- (Y4) -- (X5) -- (Y3) -- (X6)-- (Y4)-- (X7);
\draw [black] (X1) -- (Y2) -- (X4);
\draw [black] (X3) -- (Y1);
\draw [black] (X7) -- (Y3);
\draw [black] (LX1) -- (LY1)-- (LX3)--(LY3) --(LX7);
\draw [black] (LX1) -- (LY2)-- (LX2);
\draw [black] (LX2) -- (LY1);
\draw [black] (LY2) -- (LX4)-- (LY4) -- (LX6) -- (LY3) -- (LX4);
\draw [black] (LY2) -- (LX3);
\draw [black] (LY4) -- (LX7);
\draw [black] (LY4) -- (LX5) -- (LY3);
\node[left] at (LX1) {\footnotesize{$x_1$}};
\node[right] at (LX5) {\footnotesize{$x_5$}};
\node [right] at (LX6) {\footnotesize{$x_6$}};
\node [right] at (LX7) {\footnotesize{$x_7$}};
\node [below] at (LY2) {\footnotesize{$y_2$}};
\node [below] at (LX4) {\footnotesize{$x_4$}};
\node [below] at (LX3) {\footnotesize{$x_3$}};
\node [below] at (LY4) {\footnotesize{$y_4$}};
\node [above] at (LY1) {\footnotesize{$y_1$}};
\node [above] at (LX2) {\footnotesize{$x_2$}};
\node [above] at (LY3) {\footnotesize{$y_3$}};
\node [above] at (Y1) {\footnotesize{$y_1$}};
\node [above] at (Y2) {\footnotesize{$y_2$}};
\node [above] at (Y3) {\footnotesize{$y_3$}};
\node [above] at (Y4) {\footnotesize{$y_4$}};
\node [below] at (X1) {\footnotesize{$x_1$}};
\node [below] at (X2) {\footnotesize{$x_2$}};
\node [below] at (X3) {\footnotesize{$x_3$}};
\node [below] at (X4) {\footnotesize{$x_4$}};
\node [below] at (X5) {\footnotesize{$x_5$}};
\node [below] at (X6) {\footnotesize{$x_6$}};
\node [below] at (X7) {\footnotesize{$x_7$}};

\end{tikzpicture}}
\end{center}
\caption{A convex bipartite graph and its associated distance layers from $x_1$.}\label{fig:corner}
\end{figure}

\section{List $k$-Coloring on Biconvex Graphs} \label{sec:biconvex}

Enright, Stewart and Tardos \cite{enright2014} show that {\sc Li $k$-Col}, as well as  the general {\sc Li $H$-Col}, is solvable in polynomial time when restricted to graphs with all connected induced subgraphs having a multichain ordering. They apply this result to  permutation graphs and  interval graphs. Here we show that connected biconvex graphs also admit a multichain ordering.

The {\it distance layers} of a connected graph $G= (V, E)$ from a vertex $v_0$ are $L_0, L_1, . . . , L_z,$ where $L_0=\{v_0\}$ and, for $i>0$, $L_i$ consists of the vertices at distance $i$ from $v_0$ and $z$ is the largest integer for which this set is non-empty (see Figure \ref{fig:corner} for an example).  These layers form a {\it multi-chain ordering} \cite{brandstadt2003} of $G$ if, for every two consecutive layers $L_i$ and $L_{i+1}$, the edges connecting these two layers form a chain graph (not necessarily the layers themselves). All connected bipartite permutation graphs \cite{brandstadt2003} and interval graphs \cite{enright2014} admit multichain orderings. Observe that, for  the  graph given in Fig. \ref{fig:corner}, the distance layers from $x_1$ provide a multichain ordering.

Recall that a subdivision  of a graph $G$ is the graph $G'=subd(G)$ obtained from $G$ by replacing each edge by a path of length two. Thus $|E(G')|=2|E(G)|$ and $|V(G')|=|V(G)|+|E(G)|$.

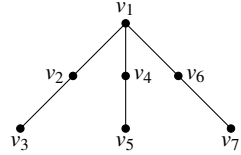
\begin{figure}[t]
\begin{center}
\begin{tikzpicture}[scale=0.7]
\draw[fill] (0,2) circle(2pt);
\foreach \i in {-1,0,1}
{
\draw[fill] (\i,1) circle(2pt);
\draw (0,2)--(\i,1);
}
\foreach \i in {-2,0,2}
{
\draw[fill] (\i,0) circle(2pt);
}
\draw (-1,1)--(-2,0) (0,1)--(0,0) (1,1)--(2,0);
\node[above] at (0,2) {$v_1$};
\node[left] at (-1,1) {$v_2$};
\node[right] at (0,1) {$v_4$};
\node[right] at (1,1) {$v_6$};
\node[below] at (-2,0) {$v_3$};
\node[below] at (0,0) {$v_5$};
\node[below] at (2,0) {$v_7$};
\end{tikzpicture}
\end{center}
\caption{Subdivision of $K_{1,3}$.}\label{fig:H}
\end{figure}

\begin{lemma}\label{claim:H}
If $G$ is a biconvex graph, then $G$ does not contain $subd(K_{1, 3})$ as an induced subgraph.
\end{lemma}

\begin{proof}
Let $G$ be a biconvex graph and let $H=subd(K_{1, 3})$. Let $v_1$ be the vertex of degree $3$ in $H$,  $v_2, v_4$ and $v_6$ be the vertices in $N(v_1)$ and  $v_3, v_5$ and $v_7$ the vertices of degree $1$ so that $v_i$ is adjacent to $v_{i+1}$ for $i=2, 4 , 6$, see Fig. \ref{fig:H}.

We observe that there is no ordering of $\{v_1,\ldots ,v_7\}$ in which the three sets $N(v_2)=\{v_1, v_3\}$, $N(v_4)=\{v_1, v_5\}$ and $N(v_6)=\{v_1, v_7\}$ become consecutive. Therefore, a bipartite graph which contains $H$ as an induced subgraph does not admit a biconvex ordering.
\end{proof}


\begin{proposition}\label{prop:multichain} Every connected biconvex graph admits a multichain ordering.
\end{proposition}

\begin{proof}

%

To see that biconvex graphs admit a multichain ordering, we use the notion of biconvex straight ordering  introduced by Abbas and Stewart \cite{stewart2000}. Let $G=(X, Y, E)$ be a bipartite graph with a linear ordering $\le $ defined on $X\cup Y$. Two edges $xy, x'y'\in E$, where $x,x'\in X$ and $y,y'\in Y$, are said to {\it cross} if $x<x'$ and $y>y'$. If $xy$ and $x'y'$ cross, we call  $(x,y')$ and $(x',y)$ the {\it corresponding straight pairs}. 
An ordering on $X\cup Y$  is a {\it straight ordering} if, for each pair $xy, x'y'$ of crossing edges, at least one of the corresponding straight pairs, $(x,y')$ or $(x',y)$,   is an edge of the graph \cite{stewart2000}. 

Let $G=(X, Y, E)$ be a connected biconvex graph. It follows from \cite[Theorem 11]{stewart2000} that $G$ admits a biconvex straight ordering, say $v_0, v_1,\ldots ,v_n$ of $X\cup Y$. Let $L_0=\{v_0\}, L_1, . . . , L_m$ be the distance layers of $G$ from $v_0$. Since the graph $G$ is connected, $V=L_0\cup L_1\cup \cdots \cup L_m$. Let us show  that these layers form a multi--chain ordering. 

The first layers $L_0$ and $L_1$ trivially form a multi--chain ordering. Let $L_1=\{v_{i_1},\cdots ,v_{i_\ell}\}$, where the vertices are listed according to the ordering. When $\ell=1$,  $L_1, L_2$ trivially form a chain graph. When $\ell>1$, since the ordering is straight, all the  edges joining $v_{i_1}$ with vertices in $L_2$ cross with the edge $v_0v_{i_2}$. As, $v_{i_2}$ is not connected to $v_0$, the other straight pair $(v_{i_1}, v_{i_2)}$ should be an edge in $G$.  Therefore $N(v_{i_1})\subseteq N(v_{i_2})$.  By iterating the same argument, we see that $N(v_{i_1})\subseteq \cdots \subseteq N(v_{i_\ell})$. Thus the layers $L_0, L_1, L_2$ form a multi--chain ordering. 

When $m=3$, let us show that $L_0, L_1,L_2,L_3$ form a multichain ordering. When $|L_3|=1$,  trivially  $L_0, L_1,L_2,L_3$ form a multichain ordering. Otherwise, assume that $L_1=\{v_{i_1},\cdots ,v_{i_\ell}\}$ and, for a contradiction, that the bipartite graph induced by $L_2\cup L_3$ contains an induced copy of $2K_2$, say with edges $uv, u'v'$ with $u,u'\in L_2$ and $v,v'\in L_3$. Since the ordering is straight, we may assume that $u<u'$ and $v<v'$. Since $u,u'\in N(v_{i_\ell})$, $v<v'$ and the ordering is biconvex, we must have $v'<v_{i_\ell}$. But then $N(u)$ contains $v, v_{i_\ell}$ but not $v'$, contradicting the biconvexity of the ordering. Thus $L_0, L_1, L_2, L_3$ form a multi--chain ordering.

Suppose that $m > 3$. 
 Let $i > 3$ be the largest subscript such that  $L_0,L_1,\ldots ,L_i$ form a multichain ordering. Suppose for a contradiction that $i<m$, Thus the bipartite graph induced by the layers $L_i, L_{i+1}$ contain an induced copy of $2K_2$, say with edges  $uv, u'v'$, $u,u'\in L_i$ and $v,v'\in L_{i+1}$.  As the ordering is straight, we may assume $u<u'$ and $v<v'$. We consider two cases:

{\it Case 1}: $N(u)\cap N(u')\cap L_{i-1}\neq\emptyset$. Let $w\in N(u)\cap N(u')\cap L_{i-1}$ and consider predecessors $w'\in L_{i-2}, w''\in L_{i-3}$ of $w$. Then the subgraph induced by $w,w',w'',u,u',v,v'$ is isomorphic to a subdivision $H$ of $K_{1,3}$, contradicting Lemma \ref{claim:H}.

{\it Case 2}: $N(u)\cap N(u')\cap L_{i-1}=\emptyset$. Let $w\in N(u)\cap L_{i-1}$ and $w'\in N(u')\cap L_{i-1}$ be some predecessors of $u$ and $u'$ in the previous layer. Observe that the two edges $wu, w'u'$ induce a $2K_2$ in the subgraph induced by  $L_{i-1}\cup L_i$ contradicting the choice of $i$. 
\end{proof}

Proposition \ref{prop:multichain} and the main result by Enright, Stewart and Tardos \cite[Theorem 2.1]{enright2014} give us our main result in this section.

\begin{theorem} For any  $H$, {\sc Li $H$-Col} is solvable in polynomial time when restricted to biconvex graphs.
\end{theorem}

As  {\sc Li $k$-Col} is a particular case of   {\sc Li $H$-Col} and {\sc Li $k$-Col} generalizes {\sc $k$-PrExt}, we have the following corollary.

\begin{corollary}\label{cor1}
{\sc Li $k$-Col} and {\sc $k$-PrExt} are solvable in polynomial time when restricted to biconvex graphs.
\end{corollary}

Concerning the running time of the algorithms, it is shown in Abbas and Stewart \cite{stewart2000} that  a biconvex straight ordering of a biconvex bipartite graph  can be found in linear time on the number of vertices of the graph. On the other hand, the algorithm in  \cite{enright2014}  is shown to run in time $O(n^{k^2-3k+4})$ time when a multichain ordering in decreasing ordering of degrees is given.  Observe that to get such ordering we have only to reorder the elements in the layers provided by the straight ordering, therefore it can be obtained in linear time.  All together gives an upper bound $O(n^{k^2-3k+4})$ on the complexity of {\sc Li $k$-Col} in the class of biconvex graphs.


\section{List $k$-Coloring of Convex Bipartite Graphs}\label{sec:convex}

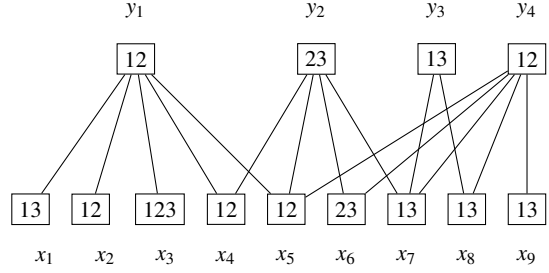
\begin{figure}[t]
\begin{center}
\begin{tikzpicture}[scale=0.8]

\draw (2.5,3.8) node {$y_1$};
\draw (5.5,3.8) node {$y_2$};
\draw (7.5,3.8) node {$y_3$};
\draw (9,3.8) node {$y_4$};

\node [draw] (y1) at (2.5,3) {12};
\node [draw] (y2) at (5.5,3) {23};
\node [draw] (y3) at (7.5,3) {13};
\node [draw] (y4) at (9,3)  {12};

\node[draw] (x1) at (0.75,0.5)  {13};
\node[draw] (x2) at (1.75,0.5)  {12};
\node[draw] (x3) at (2.9,0.5)  {123};
\node[draw] (x4) at (4,0.5)  {12};
\node[draw] (x5) at (5,0.5)  {12};
\node[draw] (x6) at (6,0.5)  {23};
\node[draw] (x7) at (7,0.5)  {13};
\node[draw] (x8) at (8,0.5)  {13};
\node[draw] (x9) at (9,0.5)  {13};

\draw (1,-0.3) node {$x_1$};
\draw (2,-0.3) node {$x_2$};
\draw (3,-0.3) node {$x_3$};
\draw (4,-0.3) node {$x_4$};
\draw (5,-0.3) node {$x_5$};
\draw (6,-0.3) node {$x_6$};
\draw (7,-0.3) node {$x_7$};
\draw (8,-0.3) node {$x_8$};
\draw (9,-0.3) node {$x_9$};

\path (y1) edge (x1);
\path (y1) edge (x2);
\path (y1) edge (x3);
\path (y1) edge (x4);
\path (y1) edge (x5);

\path (y2) edge (x4);
\path (y2) edge (x5);
\path (y2) edge (x6);
\path (y2) edge (x7);

\path (y3) edge (x7);
\path (y3) edge (x8);

\path (y4) edge (x5);
\path (y4) edge (x6);
\path (y4) edge (x7);
\path (y4) edge (x8);
\path (y4) edge (x9);

\end{tikzpicture}
\end{center}
\caption{A list assignment for  the convex bipartite graph  given in Figure \ref{fig:convex-bgraph}.  Labels inside vertex indicate the list of colors, from $\{1,2,3\}$, associated to the node.\label{fig:convex-listbgraph}}
\end{figure}

Let $G=(X\cup Y, E)$ be a connected  bipartite graph that is convex with respect to $X$.  Let $X=\{x_1,\ldots ,x_n\}$ be a convex ordering of $X$, that is, for each $y\in Y$ there are two positive integers  $a_y\leq b_y$ such that $N(y)=\{x_i\mid  a_y\le i\le b_y\}.$

Consider the set of integers $A=\{a_y\mid y\in Y\}$ and  $B=\{b_y\mid y\in Y\}$.
For the graph given in Fig. \ref{fig:convex-listbgraph}, $A=\{1,4,5,7\}$ and $B=\{5,7,8,9\}$.

We use the  set $B$ to direct  the dynamic programming  algorithm  and the elements in $A$ to determine the relevant information to be kept for the next step.
Assume that $B=\{b_1,\dots, b_\beta\}$ are sorted  so that $b_1<b_2< \dots <b_\beta$.  By connectivity of $G$, we have $b_\beta=n$.    For each $1\leq j\leq \beta$, let $X_j=\{x_i\in X\mid i\leq b_j\}$, $Y_j=\{y\in Y\mid b_y\leq b_j\}$, and  $Z_j=\{y\in Y\mid a_y\leq b_j < b_y\}$.  Define $G_j=G[X_j\cup Y_j]$. Observe that $G_\beta=G$, $Z_\beta=\emptyset$ and that $Z_j$ contains those vertices in $Y$  whose neighborhood starts before or at $b_j$ and ends after $b_j$. 
For example, for the graph given in Fig. \ref{fig:convex-listbgraph}, $b_2=7, X_2=\{x_1,x_2,...,x_7\}, Y_2=\{ y_1,y_2\}$ and $Z_2=\{y_3,y_4\}$.  For sake of simplicity we assume an initial point $b_0=0$, so that $G_0$ is the empty graph.

Let $K$ be a set of $k$ colors. Assume that each vertex $u$ in $G$ has an associated list $L(u)\subseteq K$.   We next define  the information that we want to compute for each $1\leq b_j \leq b_\beta$.  For each $1\leq j\leq \beta$, define $A(j) =\{a_y\mid y\in Z_j\}\cup \{b_j\}$. As before we assume that the elements in $A(j)=\{a_{1,j},\ldots ,a_{\alpha_j,j}\}$ are increasingly ordered, $a_{1,j}<a_{2,j}<\cdots <a_{\alpha_j,j}=b_j$. To simplify notation set  $a_{\alpha_j +1,j}=b_j+1$ to make sure that a higher value always exists. For the example in Fig. \ref{fig:convex-listbgraph}, for instance, $A(2)=\{5,6,7\}$. For the fictitious initial $b$ value $j=0$, we take $A(0)=\{0\}$

Fix $j$, $1\leq j\leq \beta$. For each   $1\leq i  \leq \alpha_j $ and  $S\subsetneq K$, $T_j(i,S)$ will hold value true whenever there is a valid list coloring of $G_j$ such that it uses no color in $S$ for the set $X_i^j=\{x_\ell\mid  a_{i,j}\leq \ell < a_{i+1,j}\}$.  Observe that we are not considering $K$ as a potential set  as not using any color is impossible. 

The Color Algorithm will compute those values in three steps. In going from $j-1$ to $j$, first it computes the values for the $x\in X_{j}$ that were not in $X_{j-1}$ combining this information with the relevant information computed in the previous step. Next, it incorporates the restriction from $y\in Y_{j}$ that were not in $Y_{j-1}$. Finally, it rearranges the information to keep only the values for the index in $A(j)$.

\begin{description}
\item[Color Algorithm:] Let $j$, $1\leq j\leq \beta$. Initially set $A(0)=\{0\}$, $b_0=0$ and set $T_0(0,S)$ to TRUE for any $S$. When $j\geq 1$ assume that the values of $T_{j-1}$ have already been computed. 
\item[Step 1]  Extending to new parts.

Let $A\rq{}(j) = A(j-1)  \cup \{a_y\mid b_{j-1}< a_y \leq b_j\} \cup \{b_j\}$. For $j>1$, by construction, those values lie before $b_{j}$ and some of them have no corresponding entries in $T_{j-1}$.  Assume that $A\rq{}(j)=\{a\rq{}_1,\dots,a\rq{}_{\gamma_j}\}$ increasingly ordered.  Let $a\rq{}_{\gamma_j+1} =b_j +1$.  We set $T_{j-1}(\ell,S)$ for $\alpha_{j-1} < \ell \leq \gamma_j$ and $S\subsetneq K$ to be true  whenever there is a valid list coloring of the set $X'(\ell)=\{x_i\mid  a\rq{}_\ell \leq i < a\rq{}_{\ell +1}\}$. For this, the algorithm checks whether $L(x)\setminus S\neq \emptyset$ for each $x\in X\rq{}(\ell)$. If this is the case, one can select a color not in $S$ and get a valid coloring. Accordingly we update the value of $T_{j-1}(\alpha_{j-1},S)$ so that it remains TRUE if it was already set to TRUE and the previous condition holds for the elements in $X'(\alpha_{j-1})$ 
\item[Step 2]  Incorporating $Y_j$.

For $y\in Y_j$ and $a_i\in [a_y,b_y]$, consider any entry $T_{j-1}(i,S)$ set to TRUE. If $S\cap L(y)=\emptyset$, the corresponding entry is changed to FALSE. 

Next, the values on $T_{j-1}$ are processed in increasing order of $x_i$: any entry $(i,S)$ holding value TRUE  will remain TRUE whenever there is an entry $(i-1, S\rq{})$ holding value TRUE with $S\subseteq S\rq{}$. By monotonicity, the property holds whenever $T_{j-1}(i-1,S)$ is TRUE. 

After processing $y$, if $T_{j-1}(l,S)$ holds true, for each  piece $[a_l', a_{l+1}')$ between $a_y$ and $b_y$, we can pick a common  color not in $S$ but in $L(y)$  to color $y$  that is compatible with some list coloring on the $X$ relevant parts that do not use $S$.

\item[Step 3] Compacting to get $T_j$. 

For each $1\leq i  \leq \alpha_j $ the set $X_i^j$  might contain several  subintervals on $X'(j)$, considered either in $T_{j-1}$ that will not be needed later on. We fusion those sets from left to right, adding one at a time,  setting  $T_j(i,S)$ to true whenever there are corresponding entries holding value true for sets $S_1$ and $S_2$ so that $S\subseteq S_1\cap S_2$.  
\end{description}

\begin{algorithm}
\caption{Color Algorithm}
\begin{algorithmic}
\REQUIRE $G=(X\cup Y, E)$ and
$\forall v\in X\cup Y, L(v) \subseteq \{1, \ldots, k\} $.
\ENSURE {\sc Li $k$-Col} for $G$ or decides that there is no such coloring.
\STATE $N(y) \leftarrow \{x_i\mid  a_y\le i\le b_y\}, A=\{a_y\mid y\in Y\}, B=\{b_y\mid y\in Y\}$.
\STATE $B \leftarrow $ ordered$(B)=\{b_1,\dots, b_\beta\}$
\STATE $A(0) \leftarrow \emptyset, b_0 = 0$
\FOR{$S \subsetneq K$}
\STATE  $T_0(0,S) \leftarrow TRUE$
\ENDFOR
\FOR{$1 \leq j \leq \beta$}
\STATE $X_j \leftarrow \{x_i\in X\mid i\leq b_j\}$ 
\STATE $Y_j \leftarrow \{y\in Y\mid b_y\leq b_j\}$ 
\STATE $Z_j \leftarrow \{y\in Y\mid a_y\leq b_j < b_y\}$ 
\STATE $A(j) \leftarrow \{a_{1, j}, \ldots, a_{\alpha_j,j}\}$ = ordered $\left( \{a_y\mid y\in Z_j\}\cup \{b_j\}\right)$
\STATE $A'(j) \leftarrow \{a'_1, \ldots a'_{\gamma_j}\}$ = ordered ($A(j-1)  \cup \{a_y\mid b_{j-1}< a_y \leq b_j\} \cup \{b_j\}$)
\STATE $a'_{\gamma_j+1} \leftarrow b_{j}+1$
\STATE $K \leftarrow \{1, \ldots, k\} $
\STATE Step 1: Expand the $T_{j-1}$ to the new values in $A'(j)$
\FOR{$\alpha_{j-1} \leq \ell \leq \gamma_j$}
\FOR{$S \subsetneq K$}
\STATE $X'(\ell) \leftarrow \{x_i | a'_\ell \leq i < a'_{\ell+1}\} $
\IF {$\ell\neq \alpha_{j-1}$}
\STATE $T_{j-1}(\ell,S)= TRUE$
\ENDIF
\FOR{$x_i \in X'(\ell)$}
\STATE $T_{j-1}(\ell, S) \leftarrow T_{j-1}(\ell, S)$ \AND ($L(x_i)\setminus S\neq \emptyset$)
\ENDFOR
\ENDFOR
\ENDFOR
\STATE Step 2A: Update the expanded $T_{j-1}$ table by considering the vertices in $Y_j$
\FOR{$y\in Y_j$}
\FOR{$i\in A'(j)$}
\WHILE{$ T_{j-1}(i, S)=TRUE$}
\IF{$S\cap L(y)=\emptyset$}
\STATE $ T_{j}(i, S) \leftarrow FALSE$
\ENDIF
\ENDWHILE
\ENDFOR
\ENDFOR
\STATE Step 2B: Second Update
\FOR{$y\in Y_j$}
\FOR{$i\in A'(j)$ with $a'_i\in (a_y,b_y]$ in increasing order}
\FOR{$S \subsetneq K$}
\STATE $T_{j-1}(i, S)\leftarrow T_{j-1}(i, S) \wedge T_{j-1}(i-1, S)$
\ENDFOR
\ENDFOR
\ENDFOR
\STATE Step 3: Computing  $T_j$ by compacting the expanded $T_{j-1}$
\STATE for $i\in A(j)$ let $f(i)$ its position in $A'(j)$
\FOR {$1 \leq i \leq \alpha_j$}
\STATE $ T_{j}(i, S) = T_{j-1}(f(i),S)$
\FOR {$f(i)<\ell<f(i+1)$}
\FOR{$S \subsetneq K$}
\IF{$\forall S_1, S_2 \subsetneq K$ with $S\subsetneq S_1 \cap S_2 \ \neg  (T_{j}(i, S_1) = TRUE\wedge T_{j-1}(\ell,S_2)= TRUE)$}
\STATE $ T_{j}(i, S) \leftarrow TRUE$
\ENDIF
\ENDFOR
\ENDFOR
\ENDFOR
\ENDFOR
\end{algorithmic}
\end{algorithm}

\newpage

To examplify the Color Algorithm consider the list assignment for the graph $G$ given in Fig. \ref{fig:convex-listbgraph}. In the Tables 1 and 2 below the value $T_0(i, S)$ is calculated for each subinterval $[a_{i, 1}, a_{i+1, 1})$ in $N(y_1)=[a_1, b_1]$ and for each non-empty proper subset $S$ of $K$.

\begin{table}[h!]
\begin{center}
\begin{tabular}{ | m{2cm} |  m{1cm}| m{1cm} |  m{1cm}| m{1cm}  |  m{1cm} | m{1cm} | } 
\hline
 &  $\{1\}$ & $\{2\}$ & $\{3\}$ & $\{1, 2\}$ & $\{1, 3\}$ & $\{2, 3\}$ \\ 
\hline
$\{x_1, x_2, x_3\}$  & T & T & T & F & F & T  \\ 
\hline
$\{x_4\}$ & T & T & T & F & T & T \\ 
\hline
$\{x_5\}$ & T & T & T & F & T & T \\ 
\hline
\end{tabular}
\end{center}
\caption{Truth values for the subintervals of $N(y_1)$ after Step 1 of the execution of the Color Algorithm.}
\label{table:1}
\end{table}

\begin{table}[h!]
\begin{center}
\begin{tabular}{ | m{2cm} |  m{1cm}| m{1cm} |  m{1cm}| m{1cm}  |  m{1cm} | m{1cm} | } 
\hline
 &  $\{1\}$ & $\{2\}$ & $\{3\}$ & $\{1, 2\}$ & $\{1, 3\}$ & $\{2, 3\}$ \\ 
\hline
$\{x_1, x_2, x_3\}$  & T & T & F & F & F & T  \\ 
\hline
$\{x_4\}$ & T & T & F & F & T & T \\ 
\hline
$\{x_5\}$ & T & T & F & F & T & T \\ 
\hline
\end{tabular}
\end{center}
\caption{Truth values for the subintervals of $N(y_1)$ after Step 2 of the execution of the Color Algorithm.}
\label{table:2}
\end{table}

\begin{lemma} Let $G=(X\cup Y, E)$ be a connected convex bipartite graph, $L$ be a color assignment for $G$. There is an $L$--coloring of  $G$   if and only if there is $S\subseteq K$ such that at the end of the execution of the Color Algorithm $T_\beta(\alpha_\beta,S)= true$.  
  \end{lemma} 
\begin{proof}
Assume that $G$ admits a list coloring.  Let $c$ be  an $L$-coloring of $G$.  For $U\subseteq X$ let $S_{U}=K\setminus c(X)$.  Observe $L$ does not use any color in $S_U$ on $U$ and  furthermore, for any $y\in Y$ so that $N(y)\cap U\neq \emptyset$, $L(y)\cap S_U\neq \emptyset$. Using this fact it follows that the entries in the tables for the corresponding sets get the value true and at the end of the algorithm   $T(\beta,\{c(x_n)\})$ will be true.

Conversely, we  can prove that the Color Algorithm correctly computes the values of $T_j$ for $1\leq j\leq \beta$.  The proof is by induction. Observe that for $j=1$ the table $R$ provides the right indices and the initialization step provides the correct values for the table on an empty graph. By induction hypothesis, we assume  that the values of $T_j$ are correctly computed. Step 1 guarantees  that the desired coloring exists when adding only the $X$ part on $G_j$ to $G_{j-1}$.  Step 2, has two parts. The first one guarantees that only those entries with sets that are compatible with the list of the vertices in $Y_j$ are still alive. The second one ensures that when combining two consecutive pieces having a common neighborhood on $Y_j$ a common set of colors (a subset) is available to color these vertices.  Finally Step 3, merge tables  for pieces that have the same $Y$ neighborhood outside $G_j$, again we need to maintain a common set of colors free for potential use on this neighbors.    
\end{proof}

Finally observe that all the running time of the color algorithm is polynomial in $|G|$ and in $2^k$. Furthemore, the  {\sc $k$-PrExt} can be polynomially  reduced to {\sc Li $k$-Col}. Therefore we get our main result.  
\begin{theorem}\label{theorem:convex}
For $k\geq 3$, {\sc Li $k$-Col} and {\sc $k$-PrExt} on convex bipartite graphs can be solved in polynomial time.  
\end{theorem}

The color algorithm can be modified to solve the {\sc Li $H$-Col} on convex bipartite graphs.  For this, the algorithm keeps track instead of the unused color on the $X$ part of the  used ones. For doing that, we have to consider some longer subdivision of the intervals in the $X$ part.  Step 2 will  check that at least one of the colors in the list of $y$ is connected to all the used colors in the $X$ part.  Step 3 is also modified as the global set of used colors will be the union.  

\begin{theorem}
For any $H$, {\sc Li $H$-Col} on convex bipartite graphs can be solved in polynomial time.  
\end{theorem} 

\section{Conclusions}\label{sec:final}

In this paper the problem posed by Huang et al. \cite{huang2015}  on the computational complexity of  the {\sc Li $3$-Col} and {\sc $3$-PrExt} on chordal bipartite graphs is addressed. A partial answer to a general version of this question is given by increasing the subclasses of chordal bipartite graphs for which polynomial time algorithms for the {\sc Li $k$-Col} are known to biconvex bipartite graphs and convex bipartite graphs. Note that the later class includes  convex bipartite graphs with bounded degree, complete bipartite graphs which have unbounded treewidth, as well as graphs with unbounded cliquewidth. Interestingly enough the second result can also be extended, with a slight  modification, to solve {\sc Li $H$-Col} for the same graph class. The paper includes another result of independent interest: any connected biconvex bipartite graph admit a multichain ordering.

On the other hand, chordal bipartite graphs form a much larger graph class. Using the terminology of \cite{scheinerman1994} it is a superfactorial graph class whereas convex bipartite graphs is a factorial graph class. Although {\sc Li $k$-Col} is hard for $k\geq 4$ when restricted to chordal bipartite graphs, finding the computational complexity of {\sc Li $3$-Col} for chordal bipartite graphs is the next natural open question.

\section*{Acknowledgments}

J. D\'{\i}az and  M. Serna are partially supported by funds from MINECO and EU FEDER under 
  grant TIN2017-86727-C2-1-R)  AGAUR project ALBCOM (2017-SGR-786) 
\"Oznur Ya\c{s}ar Diner is partially supported by the Scientific and Technological Research Council of Turkey (TUBITAK) BIDEB 2219 [grant number 1059B191802095] and by the Kadir Has University BAP [grant number 2018-BAP-08]. Oriol Serra is  supported by the Spanish Ministry of Science under project MTM2017-82166-P.

\end{document}